\documentclass[letterpaper, 10 pt, conference]{IEEEConference}
\IEEEoverridecommandlockouts
\usepackage{svg}
\usepackage{url}
\usepackage{soul}
\usepackage{tikz}
\usepackage{breqn}
\usepackage{amsmath}
\usepackage{amssymb}
\usepackage{graphicx}
\usepackage{amsfonts}
\usepackage{booktabs}
\usepackage{pgfplots}
\usepackage{algorithm}
\usepackage{subcaption}
\usepackage[noend]{algpseudocode}

\let\labelindent\relax
\usepackage{enumitem}

\newtheorem{lemma}{Lemma}
\graphicspath{{./figures/svg_to_pdf}}

\DeclareMathOperator*{\argmax}{\arg\!\max}

\newcommand{\cross}{
  \begin{tikzpicture}[scale=0.25]
    \draw (0,0) -- (1,1);
    \draw (0,1) -- (1,0);
  \end{tikzpicture}
}

\pgfplotsset{compat=1.18}
\pgfmathdeclarefunction{gauss}{2}{
  \pgfmathparse{1/(#2*sqrt(2*pi))*exp(-((x-#1)^2)/(2*#2^2))}
}

\title{\LARGE \bf
Multi-Agent Vulcan: An Information-Driven Multi-Agent Path Finding Approach
}

\author{Jake Olkin$^{*1}$, Viraj Parimi$^{*1}$ and Brian Williams $^{1}$
\thanks{This work was supported by the BP Corporation}
\thanks{$^{1}$ Computer Science and Artificial Intelligence Laboratory, Massachusetts Institute of Technology, Cambridge, MA 01239. Corresponding at {\tt\footnotesize \{jolkin,vparimi,williams\}@mit.edu}. *These authors contributed equally to the paper.}
}

\begin{document}

\maketitle
\thispagestyle{empty}
\pagestyle{empty}

\begin{abstract}
Scientists often search for phenomenon of interest while exploring new environments. Autonomous vehicles are deployed to explore such areas where human-operated vehicles would be costly or dangerous. Online control of autonomous vehicles for information-gathering is called adaptive sampling and can be framed as a Partially Observable Markov Decision Process (POMDPs) that uses information gain as its principal objective. While prior work focuses largely on single-agent scenarios, this paper confronts challenges unique to multi-agent adaptive sampling, such as avoiding redundant observations, preventing vehicle collision, and facilitating path planning under limited communication. We start with Multi-Agent Path Finding (MAPF) methods, which address collision avoidance by decomposing the multi-agent path planning problem into a series of single-agent path planning problems. We present an extension to these methods called information-driven MAPF which addresses multi-agent information gain under limited communication. First, we introduce an admissible heuristic that relaxes mutual information gain to an additive function that can be evaluated as a set of independent single agent path planning problems. Second, we extend our approach to a distributed system that is robust to limited communication. When all agents are in range, the group plans jointly to maximize information. When some agents move out of range, communicating subgroups are formed and the subgroups plan independently. Since redundant observations are less likely when vehicles are far apart, this approach only incurs a small loss in information gain, resulting in an approach that gracefully transitions from full to partial communication. We evaluate our method against other adaptive sampling strategies across various scenarios, including real-world robotic applications. Our method was able to locate up to 200\% more unique phenomena in certain scenarios, and each agent located its first unique phenomenon faster by up to 50\%.
\end{abstract}

\section{Introduction}

Adaptive sampling methods have been applied to the task of locating phenomena of interest \cite{ayton2017}. These methods frame the problem as maximizing the mutual information gain as reward in a Partially Observable Markov Decision Process (POMDP). State-of-the-art approaches address this problem as a single-agent formulation, however they do not explore multi-agent scenarios to the same fidelity. 

The multi-agent extension has additional requirements to ensure efficient sampling over the single-agent version. First, we must ensure that agents gather mutually informative observations. Multiple agents observing the same area leads to ineffective exploration when the environment is static. Second, agents must plan their paths without constant communication. Finally, agents must plan conflict-free paths to avoid collisions with each other.

In Multi-Agent Path Finding (MAPF), algorithms focus on the problem of planning conflict-free paths for multiple agents from their start locations to their goals. A fundamental strategy in MAPF involves decoupling multi-agent path planning using individual single-agent path planners, then identifying and resolving conflicts through Conflict-Based Search. While we do have a coupled, multi-agent path planning problem like MAPF, the multi-agent POMDP is coupled through the reward function, as opposed to collision conflicts. This is because the reward from each agent's observations depends on the observations from other agents. Further, MAPF techniques are ill-equipped to solve adaptive sampling problems as these techniques require additional goal specification beyond a reward function. However, we still draw inspiration from the MAPF approach by introducing a decoupled, admissible heuristic. Building from this heuristic design, we propose a method to efficiently solve the coupled multi-agent POMDP problem. We show that this heuristic guides our search over the multi-agent POMDP without the need to calculate our computationally demanding reward. This enables coordinated actions among agents to optimize collective information gain. 

Additionally, we demonstrate that, if we enforce constraints on the range of communication, we can operate in a distributed manner without the requirement for a central computing node. We achieve this by solving the coupled multi-agent planning problem whenever agents are within communication range of each other, and otherwise employ a single-agent forward search procedure that runs independent of other agents. This results in a near-optimal solution because agents often make redundant observations while they are near each other, and once two agents enter communication range, by exchanging all previous observations they will not return to areas that they had previously observed.

Current state-of-the-art approaches in multi-agent adaptive sampling do not use an information-driven POMDP formulation, which is crucial for modelling stochastic observations and the coupled nature of the reward function. Unlike multi-agent reinforcement learning-based methods which treat the reward as a deterministic function that can be calculated in a decoupled manner \cite{multi-rl-infomapf, proteins, MAAS}, the information-driven POMDP approach can account for the fact that an individual's agent's observations are only valuable if they are not redundant. Other strategies assume an intermediate model that can be updated in a decoupled manner, which makes the computation of the reward function more efficient \cite{bone2023decentralised, multi-agent-rrt, dec-pomdp}. While these mimic some types of adaptive sampling scenarios, past work has shown that the mutual information gain objective is ideal for the problem of locating phenomena of interest \cite{near-optimal-sampling, only-2-agents, single-agent-boat}. 

Alternative approaches have utilized a Monte-Carlo Tree Search (MCTS) approach to the information gathering problem, both in the single-agent case \cite{Ayton2017RiskboundedAI} and the multi-agent case \cite{bone2023decentralised}. MCTS approximates information yield by sampling potential paths and observations available to the agents. However, multiple agents introduces both a exponentially large state space, as well as multiple local minima and maxima as different routes for agents can yield similar information gains, potentially causing MCTS to miss paths closer to the optimum due to the limited sample size.

To give an overview of the rest of the paper, first in Section \ref{Adaptive Search} we define the multi-agent adaptive search problem.  In Section \ref{Methodology}, we describe our solution in two parts: the action loop for the agents, and a description of the search algorithm used to perform multi-agent search. This is found in Section \ref{multi-agent-search} which includes a proof of our approach's correctness through the lens of heuristic search. Lastly, in Section \ref{experiments} we present the experiments that we have run comparing our algorithm to similar information-driven search techniques.   

\section{Adaptive Search} 
\label{Adaptive Search}

At a high level, we address the problem of multiple autonomous agents travelling in an environment to maximize the number of detected phenomena of interest over a fixed mission duration. We assume that the agents can communicate without loss of information when they are within a range $r$ of each other or there is no communication between the agents otherwise. This limited communication paradigm divides our problem into two distinct modes: planning for the agents when they operate independently and planning for agents when they can communicate. When the agents operate outside the communication range, no additional coordination is necessary. Hence, our novelty lies in handling scenarios where an agent operates while communicating with other agent(s) in order to tackle the redundancy problem mentioned prior.

The concept of adaptive search operates under the premise that the mission duration is insufficient for a comprehensive exploration of the environment. Therefore, it is crucial for agents to gather information through measurements and utilize these findings to inform their future actions. In a multi-agent scenario with limited communication capabilities, agents should harness the measurements obtained by their counterparts to swiftly identify and disregard unpromising regions, while directing their focus towards exploring promising areas in detail. 

For the purpose of notation, we denote any random variable as $X$, with a specific value indicated by $x$. Additionally, superscript notation signifies time, while subscript notation signifies agents or locations. Therefore, $X_{i,j}^t$ represents a random variable associated with location $i$ and agent $j$ at time step $t$.

\subsection{Environment Structure}
We use similar environment structure $\mathcal{E}$ as in \cite{ayton2017}, where we model the presence of a target phenomenon at each location as a distinct discrete random variable $X_i \; \forall \; i \in [1, n]$ where $n$ is the number of distinct discrete locations in the environment. For each location, we have a random variable $U_i$, which represents whether the agent detects a feature associated with the phenomenon at location $i$. We assume that $X_i$ is conditionally independent of all phenomenons at other locations given the associated cell's feature $U_i$ thereby forming a Markov Random Field (MRF). Further, we use $Y_i$ as the noisy counterpart of the feature random variable $U_i$ and represent the underlying MRF between the features using a gaussian process $\mathcal{GP}(m(x), k(x,x'))$ where $m(x)$ is the mean function and $k(x, x')$ is the kernel function of the gaussian process. 

\subsection{MA-POMDP formulation}
Given the discretized environment structure $\mathcal{E}$ and building upon \cite{ayton2017}, we formulate our problem as a discrete finite-horizon POMDP $M_j$ for each agent $a_j$. This POMDP is defined by a 8-tuple $\{\mathcal{S}_j, \mathcal{A}_j, \mathcal{T}_j, \Omega_j, \mathcal{O}_j, R, \gamma, \delta\}$. Each observation taken by agent $a_j$ at location $i$ up to time step $t$ is denoted as $y_j^{0:t}$. The state space of agent $j$, $s_j \in \mathcal{S}_j$, is formed by combining the observation, feature probability function ($\lbrace p(u_i \mid y_j^{0:t}) \rbrace_{i=1}^{n}$), and phenomenon probability function ($\lbrace p(x_i \mid y_j^{0:t}) \rbrace_{i=1}^{n}$). The action space $\mathcal{A}_j$ consists of discrete movements such as up, down, left, right, or idle at the current location. However, given the objective of exploring the environment within a limited mission duration, we only consider the idle action when other actions are infeasible. Additionally, $\mathcal{T}_j: \mathcal{S}_j \times \mathcal{A}_j \rightarrow \mathcal{S}_j$ represents a deterministic transition function. The observation space $\Omega_j$ is continuous, while $\mathcal{O}_j : \mathcal{S}_j \times \mathcal{A}_j \times \Omega_j \rightarrow [0, 1]$ represents the observation probability function which is defined as the distribution over $Y_i$ at a location $i$. The reward function $R$ is designed to maximize information gain. Further, since we are dealing with a finite-horizon mission, we set $\gamma = 1$ to emphasize the importance of identifying phenomena of interest throughout the planning horizon $\delta$. Finally, for the sake of brevity we will represent the timesteps associated with the planning horizon $\delta$ i.e $t+1:t+\delta$ be represented by $\tau$.

\subsection{Reward Function}
Our reward function is inspired by \cite{ayton2017} Sections 4.4 and 4.5 and is defined as,

\begin{align*}
    R &= I(\lbrace X_i \rbrace_{i=1}^n ; Y_j^{\tau} \mid y_j^{0:t})
    \approx \sum_{i=1}^n I(X_i ; Y_j^{\tau} \mid y_j^{0:t})
\end{align*}
for an agent $a_j$ over a planning horizon $\delta$. The information objective is defined as, 

\begin{align*}
    I(X_i ; Y_j^{\tau} \mid y_j^{0:t}) &= \mathbb{E}_{Y_j^{\tau}} \left[ D_{\text{KL}}(p_{X_i \mid Y_j^{\tau}, y_j^{0:t}} \Vert p_{X_i \mid y_j^{0:t}}) \right]
\end{align*}

The phenomenon probability function is defined as, 

\begin{align*}
    p(X_i = 1 \mid y_j^{0:t}) &= \frac{P_1}{2} \left( 1 - \text{erf}\left(\frac{\tilde{u} - \mu}{\sqrt{2\Sigma}} \right) \right) \\
    &+ \frac{P_2}{2} \left( 1 + \text{erf}\left(\frac{\tilde{u} - \mu}{\sqrt{2 \Sigma}} \right) \right)
\end{align*}

where $P_1, P_2$ and $\tilde{u}$ are user-defined parameters. Here $\mu$ and $\Sigma$ define the posterior distribution of $U_i$ given the observation history $y_j^{0:t}$ as governed by the underlying gaussian process forming the feature probability function. To compute the expectations we use the 5th order Gauss-Hermite quadrature.

To give intuition about the shape and effect of the parameters of the phenomenon probability function, $\tilde{u}$ is treated as a threshold for how confident we must be in our observation of the phenomenon to treat the phenomenon as more likely present than not, and then $P_1$ and $P_2$ are weights assigned to give more credence to the case when the measurement value is above or below the threshold respectively. This is visualized in Figure \ref{fig:reward-intuition} from \cite{ayton2017}.

\begin{figure}
    \centering
    \includegraphics[width=\linewidth]{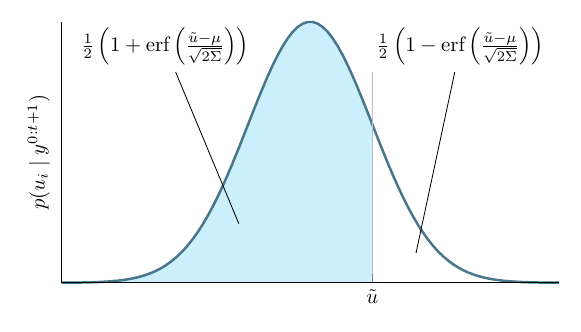}
    \caption{Probability phenomenon function for $P_1 = P_2 = 0.5$ to demonstrate the effect of $\tilde{u}$ on the posterior distribution}
    \label{fig:reward-intuition}
\end{figure}

This reward function uses the mutual information between the random variables associated with the target phenomenon of interest and the observations made by the agents. That is to say, we are incentivizing the agents to take actions so that resultant observations raises their localization confidence of the target phenomena. 

\section{Methodology} 
\label{Methodology}

Our approach as outlined in Algorithm \ref{algo:overview} builds on the idea of distributed and online execution. Given an environment structure $\mathcal{E}$, a set of agents $A$ where each agent $a_j$ is governed by their independent POMDP $M_j$ and time $t$ since the mission began, our approach first identifies the set of agents that are nearby a specific agent This implies that for each agent $a_i$ we identify a subset of agents $N_i$ who lie within the communication range $r$ by computing the Manhattan distance $d$ between their current positions (lines 2-3). However, note that an agent can be part of multiple other agents' neighborhood set. To circumvent the problem of duplicating path planning efforts for such agents, we identify the minimal disjoint sets given the agent neighborhoods (lines 4-5) where each subset $\lambda_k$ represents the agents that are within the communication range $r$ of each other. Consequently, for each minimal disjoint set $\lambda_k$ we instantiate a \textproc{Multi-Agent Search} that implements A* search over the joint state space of the agents within that set (lines 5-6). For the agents that are not in communication range $r$ of any other agent in the map, we leverage a forward search idea inspired by \cite{Ayton2017RiskboundedAI} to plan their paths (lines 7-8). After extracting the immediate actions of all agents, we execute them and collect new observations to inform future planning efforts (lines 9-10). Note that over time, $\Lambda$ evolves which implies that if $\lambda_k = \{a_1, a_2\}$ at timestep $t$, the agents $a_1$ and $a_2$ may drift apart in the next timestep $t+1$ collapsing $\lambda_k$. In such a scenario, agent $a_1$ and $a_2$ will perform single-agent search from timestep $t+1$ wherein each agent can utilize the other agents' observations up until timestep $t$. However, subsequent observations from timestep $t+1$ will not be shared between the agents $a_1$ and $a_2$ as they would be expected to plan their paths independently.

\begin{algorithm}
    \caption{High-level overview of the approach}
    \hspace*{\algorithmicindent} \textbf{Input} 
        \begin{tabular}[t]{ll}
          Environment $\mathcal{E}$, Agents $A = \{a_1, \ldots, a_k\}$ \\
          Mission Duration $H$, Communication Range $r$
        \end{tabular} \qquad
    
    \begin{algorithmic}[1]
        \While {$t \leq H$}
            \ForAll {$a_i \in A$}
                \State $N_i \gets \{a_j \mid d(a_i, a_j) \leq r \}$
            \EndFor
            \State $\Lambda \gets $ Extract minimal disjoint sets from \par
                   \hskip\algorithmicindent $\{N_i \mid i \in \{1, \ldots k\}\}$
            \ForAll {$\lambda_k \in \Lambda$}
                \State $\Pi_{\lambda_k} \gets $ \textproc{Multi-Agent Search}$(\lambda_k, \mathcal{E})$
            \EndFor
            \ForAll {$a_k \notin \Lambda$}
                \State $\Pi_{a_k} \gets $ \textproc{Single-Agent Search}$(a_k, \mathcal{E})$
            \EndFor
            \ForAll {$a_i \in A$}
                \State Execute $\Pi_{a_i}$ and collect observation $\omega_{a_i}$
            \EndFor
            \State $t \gets t + 1$
        \EndWhile
    \end{algorithmic}
    \label{algo:overview}
\end{algorithm}

\subsection{Multi-Agent Search}
\label{multi-agent-search}

When two or more agents come within communication range of each other, we form a corresponding \textit{bubble} $\lambda_k$. An agent $a_k \in \lambda_k$ becomes the lead actor who instantiates the multi-agent search process for the agents in the \textit{bubble}. This process performs an informed A* search over the joint state space of these agents and generates viable actions for each of them. To do this, we form a new state $\tilde{s} = \{ \tilde{y}^{0:t}, \lbrace p(u_i \mid \tilde{y}^{0:t}) \rbrace_{i=1}^{n},  \lbrace p(x_i \mid \tilde{y}^{0:t}) \rbrace_{i=1}^{n} \}$ where $\tilde{y}^{0:t} = \{y_j^{0:t} \mid a_j \in \lambda_k \}$ represents the combined observation history of the agents in the \textit{bubble}. Further, the action space for this search process is represented by $\cross_{a_j \in \lambda_k} \mathcal{A}_j$. The frontier states of our A* search are ordered by the priority function $f(s) = g(s) + h(s)$. Here $g(s)$ represents the expected information gain between the phenomenon of interest and the joint distribution of the observations made by the agents in the \textit{bubble} up to the planning horizon $\delta$ conditioned on the combined observation history of these agents. More specifically let $Y_{\lambda_k}^{\tau}$ represent the random variables associated with the joint distribution of the observations made by the agents up to the planning horizon, then the $g$-function is defined as follows:
\begin{dmath*}
    g(s) = \sum_{i=1}^n I{(X_i; Y_{\lambda_k}^{\tau} \mid \tilde{y}^{0:t})} 
\end{dmath*}
    
$h(s)$ represents the optimistic admissible heuristic function defined as the sum of the maximum mutual information gain between the phenomenon of interest and the distribution of observations made by that agent acting independently up to the planning horizon $\delta$ conditioned on the combined observation history of the agents in the \textit{bubble}. More specifically, let 
\begin{dmath*}
    h_j^{\tau} = \sum_{i=1}^n I{(X_i; Y_j^{\tau} \mid \tilde{y}^{0:t})} 
\end{dmath*}
represent the expected mutual information gain between the phenomenon of interest and the distribution over observations taken by agent $a_j$ up to the planning horizon, then 
\[
    h(s) = \sum_{a_j \in \lambda_k} h_j^{\tau}
\]

To ensure that the A* search returns optimal paths for the agents within $\lambda_k$, we need to ensure that the $h$-function is admissible. Since we are interested in a receding horizon plan, we search over a tree where an admissible heuristic is sufficient to ensure optimality of the A* search process. Our heuristic is provably admissible for environments where it is known that the target phenomenon $X_i$ is a direct cause of the (noisy) observation $Y_i$. To demonstrate that our heuristic is admissible, we must show that our heuristic is an optimistic estimate of the reward we could receive starting from any given state. Or, more specifically for our scenario, we must show that the maximum, multi-agent information gain from a given state will always be less than or equal to our heuristic estimate for that state. 

\begin{lemma}
Given the definitions of $g(s)$ and $h(s)$, $h(s) \geq g(s)$    
\end{lemma}
\begin{proof}
    Expanding the $g$-function for a given \textit{bubble} $\lambda_k$,
    \begin{dgroup*}
        \begin{dmath*} 
            g(s) = \sum_{i=1}^n I{(X_i; Y_{\lambda_k}^{\tau} \mid \tilde{y}^{0:t})}
            = \sum_{i=1}^n \left[ I{(X_i; Y_{\lambda_{k, 1}}^{\tau} \mid \tilde{y}^{0:t})} + I{(X_i; Y_{\lambda_{k, 2}}^{\tau} \mid Y_{\lambda_{k, 1}}^{\tau}, \tilde{y}^{0:t})} + \ldots + I{(X_i; Y_{\lambda_{k, m}}^{\tau} \mid (Y_{\lambda_{k, 1}}^{\tau}, \ldots Y_{\lambda_{k, m-1}}^{\tau}), \tilde{y}^{0:t})} \right]
            = h_1^{\tau} + \sum_{i=1}^n \left[ I(X_i; Y_{\lambda_{k, 2}}^{\tau} | Y_{\lambda_{k, 1}}^{\tau}, \tilde{y}^{0:t}) + \ldots + I(X_i; Y_{\lambda_{k, m}}^{\tau} | (Y_{\lambda_{k, 1}}^{\tau}, \ldots Y_{\lambda_{k, m-1}}^{\tau}), \tilde{y}^{0:t}) \right]
        \end{dmath*}
    \end{dgroup*}
Here, $Y_{\lambda_{k, j}}^{\tau}$ represents the random variable associated with the observations of agent $a_j \in \lambda_k$ up to the planning horizon $\delta$ and $m = |\lambda_k|$. Comparing terms between this and the $h$-function defined earlier, we observe that the first term cancels out. To establish the required relationship between the $g$ and $h$-functions, it is enough to  show that $\forall j \in [2, m], h_j^{\tau} \geq \sum_{i=1}^n I(X_i; Y_{\lambda_{k, j}}^{\tau} \mid (Y_{\lambda_{k, 1}}^{t+1} ... Y_{\lambda_{k, j-1}}^{t+1}), \tilde{y}^{0:t})$. Without loss of generality, for $j=2$ we need to show that $h_2^{\tau} \geq \sum_{i=1}^n I(X_i; Y_{\lambda_{k, 2}}^{\tau} \mid Y_{\lambda_{k, 1}}^{\tau}, \tilde{y}^{0:t})$. Examining $I(X_i; Y_{\lambda_{k, 2}}^{\tau} \mid Y_{\lambda_{k, 1}}^{\tau})$ while omitting $\tilde{y}^{0:t}$ for brevity we observe that,

\begin{dgroup*}
    \begin{dmath*}
        I{(X_i ; Y_{\lambda_{k, 2}}^{\tau} \mid Y_{\lambda_{k, 1}}^{\tau})} = I{(Y_{\lambda_{k, 2}}^{\tau}; X_i \mid Y_{\lambda_{k, 1}}^{\tau})}
        = I{(X_i ; Y_{\lambda_{k, 2}}^{\tau})} - I{(Y_{\lambda_{k, 2}}^{\tau}; Y_{\lambda_{k, 1}}^{\tau})} + I{(Y_{\lambda_{k, 2}}^{\tau}; Y_{\lambda_{k, 1}}^{\tau} \mid X_i)} 
    \end{dmath*}
\end{dgroup*}
where we use the symmetry of mutual information along with the chain-rule for conditional mutual information. This implies that,
\begin{align*}
    I(X_i ; Y_{\lambda_{k, 2}}^{\tau}) &= I{(X_i; Y_{\lambda_{k, 2}}^{\tau} \mid Y_{\lambda_{k, 1}}^{\tau})} \\
    &+ I{(Y_{\lambda_{k, 2}}^{\tau}; Y_{\lambda_{k, 1}}^{\tau})} - I{(Y_{\lambda_{k, 2}}^{\tau}; Y_{\lambda_{k, 1}}^{\tau} \mid X_i)}
\end{align*}

Under the MRF describing our environment structure, we model $Y_{\lambda_{k, 1}}^{\tau}$ and $Y_{\lambda_{k, 2}}^{\tau}$ as being caused by the respective phenomenon random variables present at their respective observation locations. These observations are collected nearby each other by virtue of the agents being within the communication radius which implies that they will correlate with $X_i$. Based on these criteria we can conclude that $X_i$ can be considered as the common cause of $Y_{\lambda_{k, 1}}^{\tau}$ and $Y_{\lambda_{k, 2}}^{\tau}$. Since we know that mutual information between two random variables $P$ and $R$ decreases when it is conditioned on another random variable $Q$ where $Q$ is the common cause of both $P$ and $R$, we can state that $I(Y_{\lambda_{k, 2}}^{\tau}; Y_{\lambda_{k, 1}}^{\tau}) \geq I(Y_{\lambda_{k, 2}}^{\tau} ; Y_{\lambda_{k, 1}}^{\tau} \mid X_i)$. Thus, this means that $I(Y_{\lambda_{k, 2}}^{\tau}; Y_{\lambda_{k, 1}}^{\tau}) - I(Y_{\lambda_{k, 2}}^{\tau}; Y_{\lambda_{k, 1}}^{\tau} \mid X_i) \geq 0$, and therefore $I(X_i; Y_{\lambda_{k, 2}}^{\tau} \mid Y_{\lambda_{k, 1}}^{\tau}) \leq I(X_i; Y_{\lambda_{k, 2}}^{\tau})$. This proves our intermediate objective of showing $h_2^{\tau} \geq \sum_{i=1}^n I(X_i; Y_{\lambda_{k, 2}}^{\tau} \mid Y_{\lambda_{k, 1}}^{\tau}, \tilde{y}^{0:t})$ where $h_2^{\tau} = \sum_{i=1}^n I(X_i ; Y_2^{\tau} \mid \tilde{y}^{0:t})$. Note that $\lambda_{k, 2}^{\tau} = \lambda_2^{\tau}$. Extending this reasoning over $j \in [2, m]$ we can see that $h(s) \geq g(s)$.
\end{proof}

With the admissibility of our heuristic function, we ensure optimal path generation for the agents inside a \textit{bubble} $\lambda_k$. However, the key challenge we encounter is that computation of the multi-agent information gain $g$ is very compute intensive because it requires iterating through all different combinations of potential observations for agents inside \textit{bubble} over the planning horizon $\delta$. Additionally, as our actions space grows exponentially with the number of agents inside the \textit{bubble}, we generate a larger number of states for every search state that we choose to expand. To mitigate this issue, we leverage our optimistic heuristic computations to ignore states that will never be expanded. 

To compute $h(s)$ for any given state $s$, we compute the maximum information gain we can receive from taking each action from $s$, effectively computing $h(c) \; \forall c \in \mathcal{C}$ where $\mathcal{C}$ is the set of the children that one can reach from the state $s$ according to the $\cross_{a_j \in \lambda_k} \mathcal{A}_j$. We use these $h(c)$ values to order the children in $s$ for generation. Let $\tilde{\mathcal{C}} \subseteq \mathcal{C}$ be the set of states for which we have computed $g(c)$. We continue to calculate $g(c)$ for $\argmax_{ c \in \mathcal{C} \setminus \tilde{\mathcal{C}}} h(c)$ until $g(s) + \max_{ c \in \mathcal{C} \setminus \tilde{\mathcal{C}}} h(c) < \max_{c \in \tilde{\mathcal{C}}} g(c)$. At this point, we know that the maximum information gain we could receive from the remaining states will be less than the information gain we can guarantee from taking a different action, and therefore will never need to be expanded. Using this observation we can reduce the number of times we compute $g(s)$ for any given state $s$. 

The \textproc{Multi-Agent Search} algorithm presented in Algorithm \ref{algo:multi} operates by starting with a \textit{bubble} $\lambda_k$ and the environment $\mathcal{E}$. It aggregates observations within the \textit{bubble} to form a new state $\tilde{s}$ by concatenating the observations from all agents within the bubble, computing both $g$ and $h$ values for it (lines 1-2). This state is added to an open list $\mathcal{Q}$, and variables for tracking the highest information gain $I^*$ and corresponding best actions $\pi^*$ are initialized (lines 3-4). The algorithm proceeds in an A* manner, selecting and removing states from the open list based on their $f$-value, aiming to maximize information gain (line 6). If a state's $f$-value doesn't surpass the current maximum gain $I^*$, the algorithm concludes, returning the optimal actions $\pi^*$ found (lines 7-8). Otherwise, for states at the planning horizon $\delta$ with higher $f$-values, it updates the maximum gain $I^*$ and actions $\pi^*$ (lines 9-11). For states not at the planning horizon $\delta$, it evaluates their descendants in order. Note that, when computing the $h$-value of a state, we also compute the information gain received from every child reachable from that state which allows us to order these children (line 13). If a descendant's \textit{optimistic} $f$-value is higher than the current maximum gain (line 15), the descendant is added to the open list for further consideration (lines 16-17). Finally we repeat this process until the open list is exhausted. Note that $t(s)$ returns the timestep of the state $s$.

\begin{algorithm}
    \caption{\textproc{Multi-Agent Search}}
    \label{algo:multi}
    \hspace*{\algorithmicindent} \textbf{Input} 
        \begin{tabular}[t]{ll}
          Agent Bubble $\lambda_k = \{a_1, ... a_m \mid d(a_i, a_j) < r\}$ \\ Environment $\mathcal{E}$
        \end{tabular} \qquad
    \begin{algorithmic}[1]
        \State $\tilde{y}^{0:t} \gets \{y^{0:t}_{a_i} \mid a_i \in \lambda_k\}$
        \State $\tilde{s} \gets \{ \tilde{y}^{0:t}, \lbrace p(u_i \mid \tilde{y}^{0:t}) \rbrace_{i=1}^{n},  \lbrace p(x_i \mid \tilde{y}^{0:t}) \rbrace_{i=1}^{n} \}$
        \State $\mathcal{Q} \gets \tilde{s}$
        \State Initialize $I^*$ and $\pi*$
        \While {$\mathcal{Q} \neq \varnothing$}
            \State $\tilde{s} \gets \argmax_{\tilde{s} \in \mathcal{Q}}(f(\tilde{s}))$
            \If {$f(\tilde{s}) \leq I^*$}
                \State Return $\pi^*$            
            \ElsIf {$t(\tilde{s}) \geq \delta \And I^* \leq f(\tilde{s})$}
                \State $I^* \gets f(\tilde{s})$
                \State Update $\pi^*$
            \ElsIf {$t(\tilde{s}) < \delta$}    
                \ForAll{ordered children $c$ of $\tilde{s} \in \cross_{a_j \in \lambda_k} \mathcal{A}_j$}
                    \State $t' \gets t(\tilde{s}) + 1$
                    \If {$I^* \leq g(\tilde{s}) + h(c)$}
                            \State $\tilde{c} \gets \{ \tilde{y}^{0:t'}, \lbrace p(u_i \mid \tilde{y}^{0:t'}) \rbrace_{i=1}^{n},$ \newline \hspace*{6.1em} $ \lbrace p(x_i \mid \tilde{y}^{0:t'}) \rbrace_{i=1}^{n} \}$
                            \State $\mathcal{Q} \gets \mathcal{Q} \cup \{ \tilde{c} \}$            
                    \EndIf
                \EndFor
            \EndIf
        \EndWhile
    \end{algorithmic}
\end{algorithm}

\section{Experiments}
\label{experiments}

Our experiments\footnote{Code is available at https://gitlab.com/mit-mers/info-mapf-public.git} address the following questions to evaluate the effectiveness of our approach. 

\begin{itemize}[noitemsep, align=left]
\item[\textbf{Q1:}] Can our method demonstrate better performance in identifying the number of phenomena of interest compared to other methods?
\item[\textbf{Q2:}] Does the proposed approach address the issue of redundant observations effectively?
\item[\textbf{Q3:}] Does the suggested approach effectively utilize the proposed heuristic to avoid the computational complexities involved in computing coupled reward function?
\end{itemize}

As there are currently no established state-of-the-art approaches for addressing information-guided MAPF, we evaluate our approach against ablations in addition to a more involved approach that leverages MCTS. Specifically, we compare our method against Single-Agent Vulcan (SA-V), an approach where each agent plans its individual paths based on the algorithm outlined in \cite{ayton2017}, along with a derived version called Single-Agent Vulcan with Collision Avoidance (SA-V-CA) that includes the necessary collision avoidance check. Finally, we compare our approach with an MCTS-based variant of our proposed approach (MA-MCTS-V) where we estimate the reward from different actions based on random rollouts. Similar to Algorithm \ref{algo:overview}, it performs an MCTS-based search as opposed to using Algorithm \ref{algo:multi} to estimate the value of different branches of the search tree and takes decisions based on those evaluations. This attempts to maximize the multi-agent information gain directly, unlike the first two algorithms which reason over single-agent information gain.

The experiments utilized established MAPF benchmarks\footnote{https://movingai.com/benchmarks/mapf/index.html} \cite{mapf_benchmarks} and were extended to include two real-world scenarios derived from bathymetric maps. Tests were performed on the standard empty 16x16, empty 32x32, maze 32x32, and dense 65x81 maps, alongside real-life scenarios in East Boston Harbor and Galveston Bay based on NOAA surveys H10992 and H10638, respectively. In East Boston Harbor, the AUV navigated at a consistent depth of 15 meters, using 15-meter depth contours as obstacles. Similarly, in Galveston Bay, 2-meter depth contours determined obstacle boundaries for the AUV. Selected map examples are displayed in Figure \ref{fig:maps}.

Our experiments comprised 100 test runs each, featuring randomly positioned agents denoted by $|A|$ and simulated measurement fields containing up to $N$ target phenomena. These tests varied in mission duration $H$ and employed a planning horizon ($\delta$) of 2 and a communication range ($r$) of 5. For the simulated fields, we used unit mean functions and a kernel function $k(x, x') = \theta_1\exp{- (\Vert x - x' \Vert_2^2 / \theta_2^2)}$ to define the Gaussian Process ($\mathcal{GP}$). We adapted the realistic scenarios to our discrete action space; East Boston Harbor was discretized to cells of $0.0003^\circ$ in latitude and longitude, translating to a 25m step movement. Galveston Bay's discretization was set at $0.001^\circ$ per cell, yielding a step movement of approximately 100m. Information gain calculations for our experiments used parameters $\tilde{u} = 1.4$, $P_1 = 0.98$, $P_2 = 0.002$, $\sigma = 0.2$, with $\theta_1 = 0.4$ and $\theta_2 = 0.01$ for MAPF benchmarks, and $\theta_1 = 1.25$, $\theta_2 = 4*d^{\circ}$ for realistic scenarios, where $d^{\circ}$ indicates the cell size. Parameters $\tilde{u}, P_1, P_2$ are required for calculating the phenomenon probability function, while $\sigma$ accounted for the measurement noise in $\mathcal{GP}$.

\begin{figure}
    \centering
    \captionsetup[subfigure]{justification=centering}
     \begin{subfigure}[b]{0.1\textwidth}
         \includegraphics[width=\textwidth]{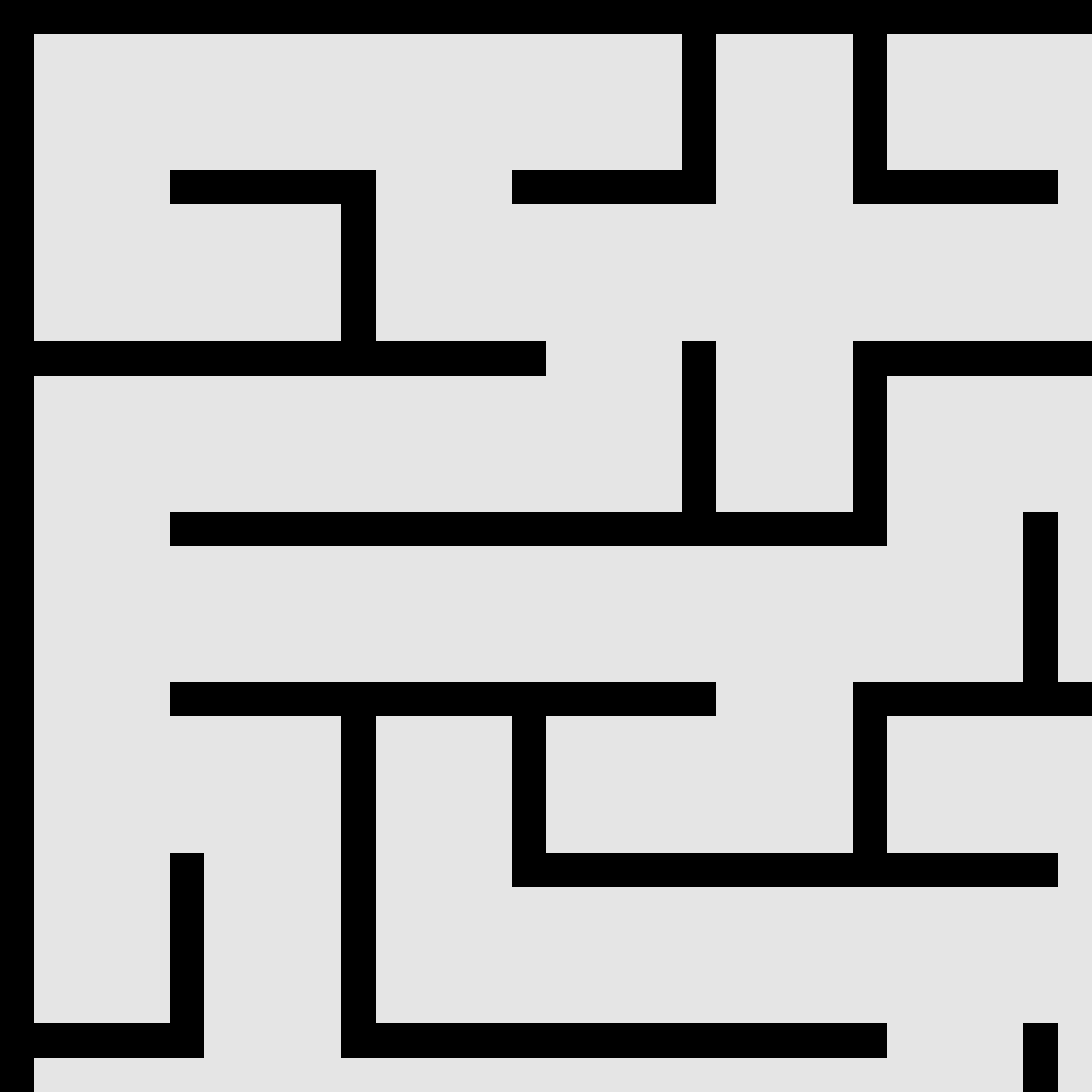}
         \caption{Maze (32x32)}
     \end{subfigure}
     \hfill
     \begin{subfigure}[b]{0.1\textwidth}
         \centering
         \includegraphics[width=\textwidth]{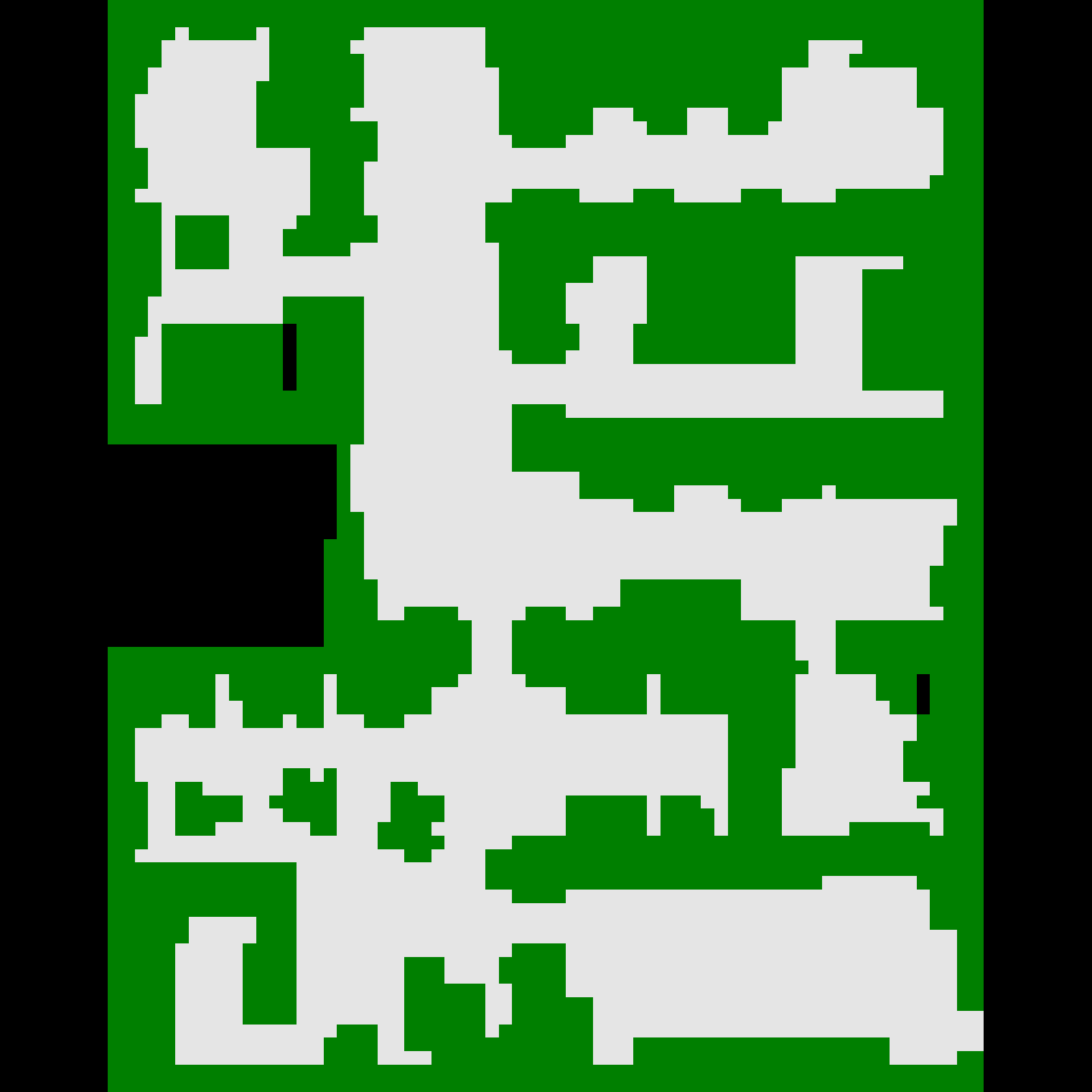}
         \caption{Dense (65x81)}
     \end{subfigure}
     \hfill
     \begin{subfigure}[b]{0.115\textwidth}
         \centering
         \includegraphics[width=\textwidth]{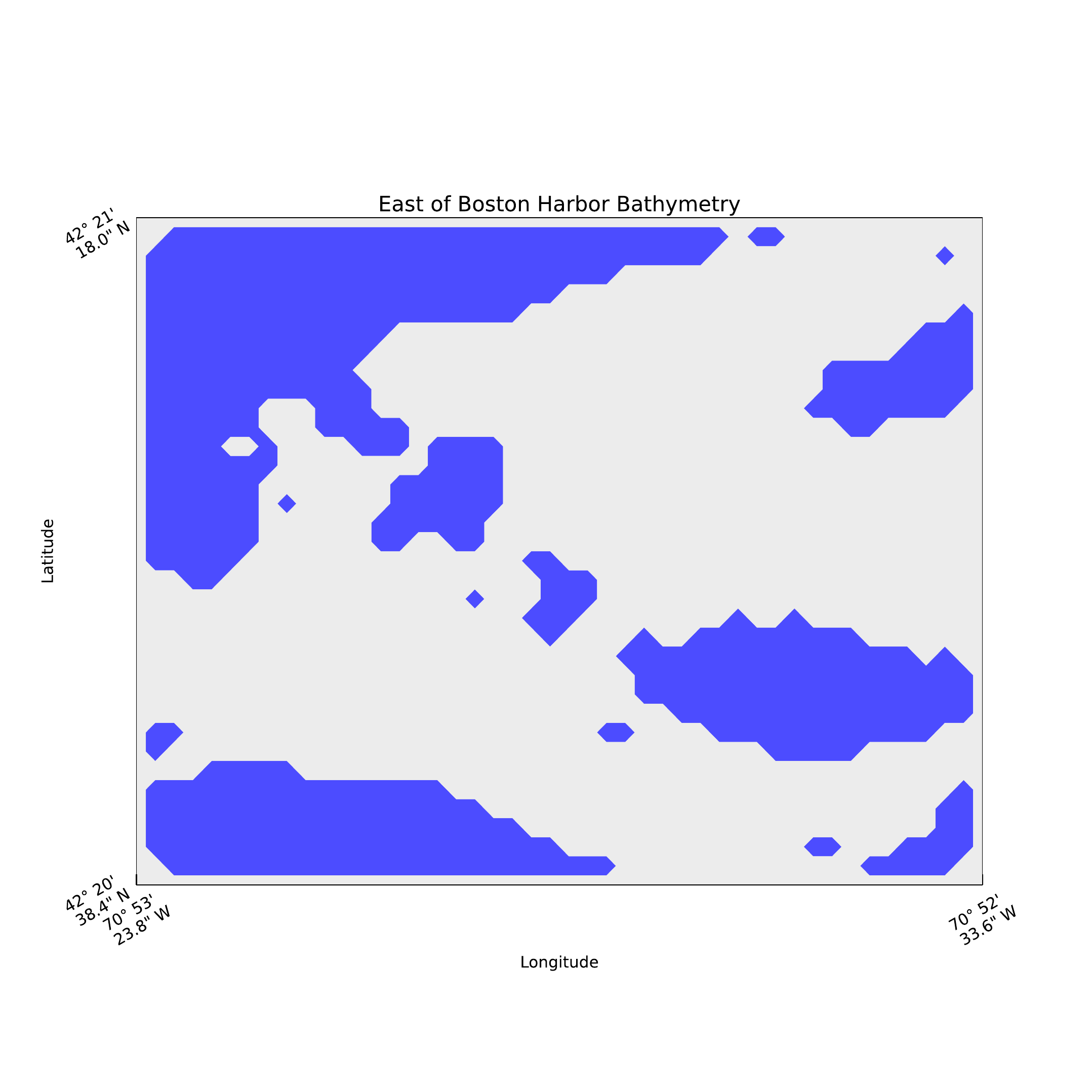}
         \caption{Boston Harbor}
     \end{subfigure}
    \begin{subfigure}[b]{0.115\textwidth}
         \centering
         \includegraphics[width=\textwidth]{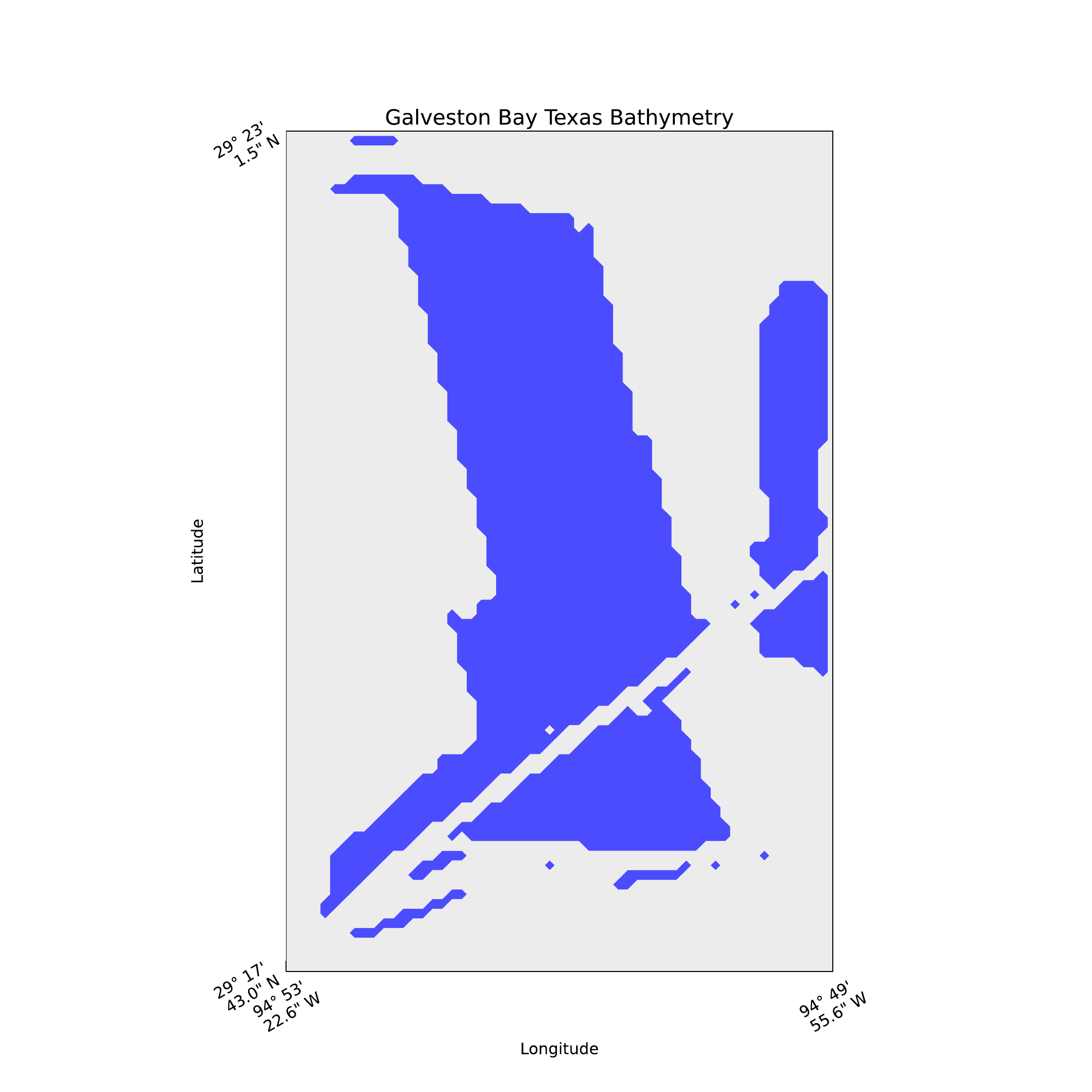}
         \caption{Galveston Bay}
     \end{subfigure} 
     \caption{Visualizations of MAPF Maps and Realistic Scenarios}
    \label{fig:maps}
\end{figure}

\begin{figure}[ht]
    \centering  
    \includegraphics[width=0.55\textwidth]{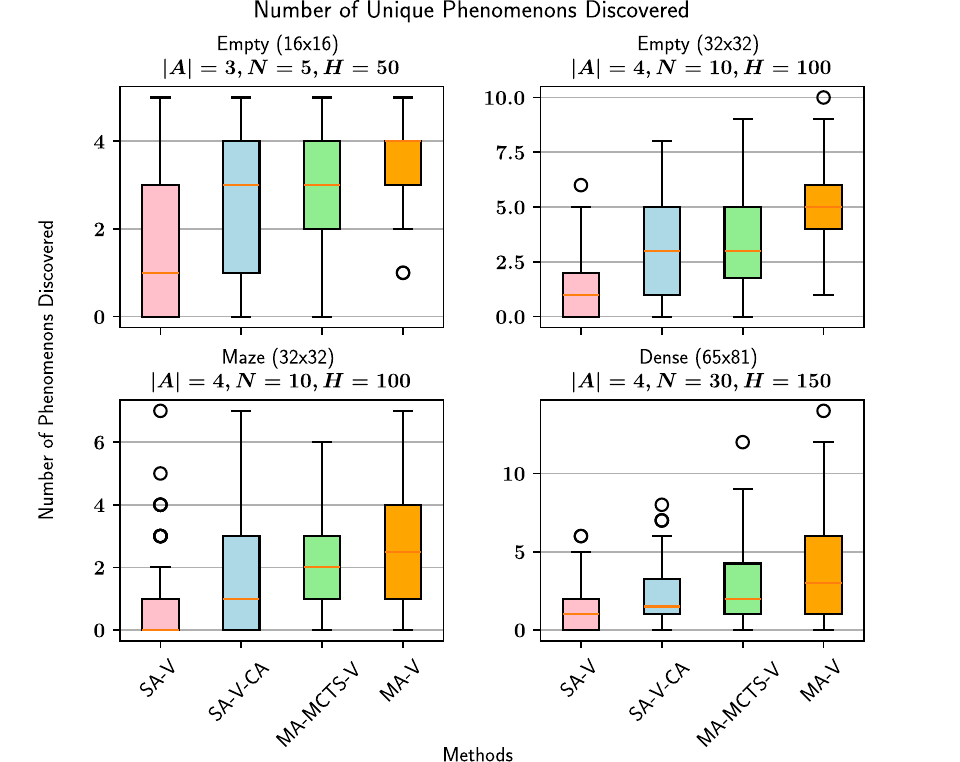}
    \caption{Total number of unique phenomena discovered by all agents on the MAPF maps. On average, our algorithm locates more phenomenon across all maps.}
    \label{fig:phenomenons_discovered_mapf}
\end{figure}

\begin{figure}[ht]
    \centering  
    \includegraphics[width=0.5\textwidth]{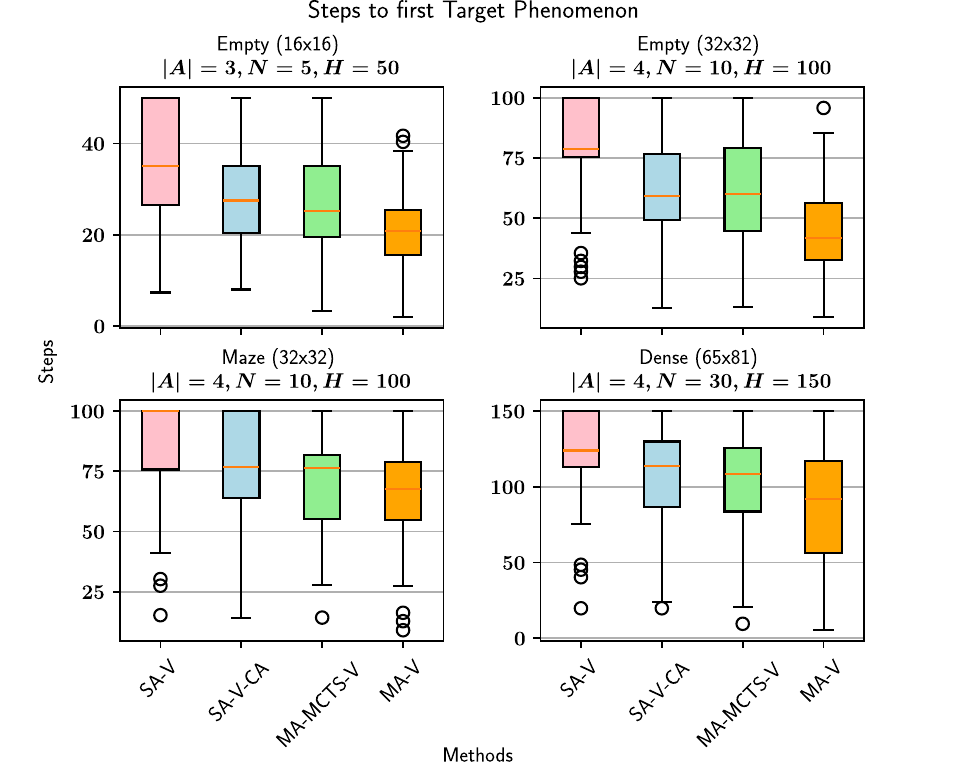}
    \caption{Average number of steps until each agent locates its first unique phenomenon on the MAPF maps. On average, each agent finds new phenomena faster using our algorithm.}
    \label{fig:step_to_first_gp_mapf}
\end{figure}

\begin{figure}[ht]
    \centering  
    \includegraphics[width=0.5\textwidth]{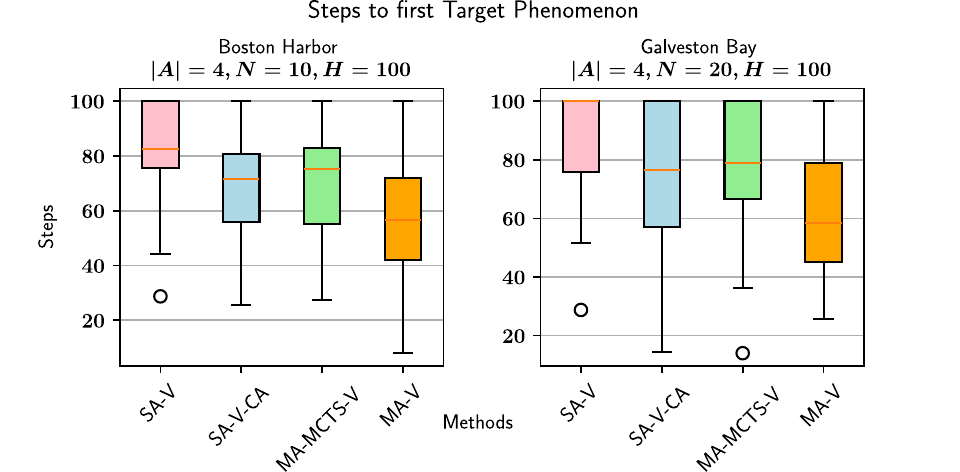}
    \caption{Average number of steps until each agent locates its first unique phenomenon on real bathymetry datasets. On average, each agent finds unique phenomena faster using our algorithm.}
    \label{fig:step_to_first_gp_real}
\end{figure}

\begin{figure}[ht]
    \centering  
    \includegraphics[width=0.5\textwidth]{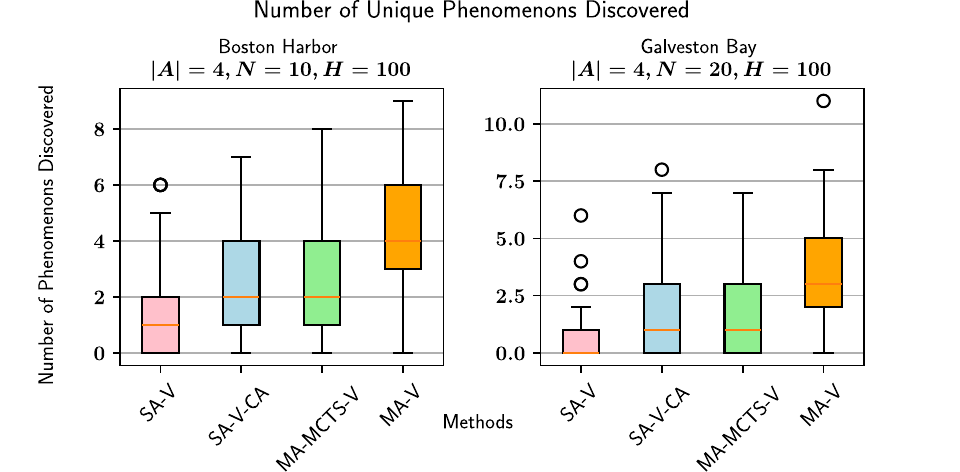}
    \caption{Total number of unique phenomena discovered by all agents on real bathymetry datasets. On average, our algorithm locates more phenomenon across all maps.}
    \label{fig:phenomenons_discovered_real}
\end{figure}

\begin{figure}[ht]
    \centering
    \scalebox{1}{
    \includegraphics[width=0.55\textwidth]{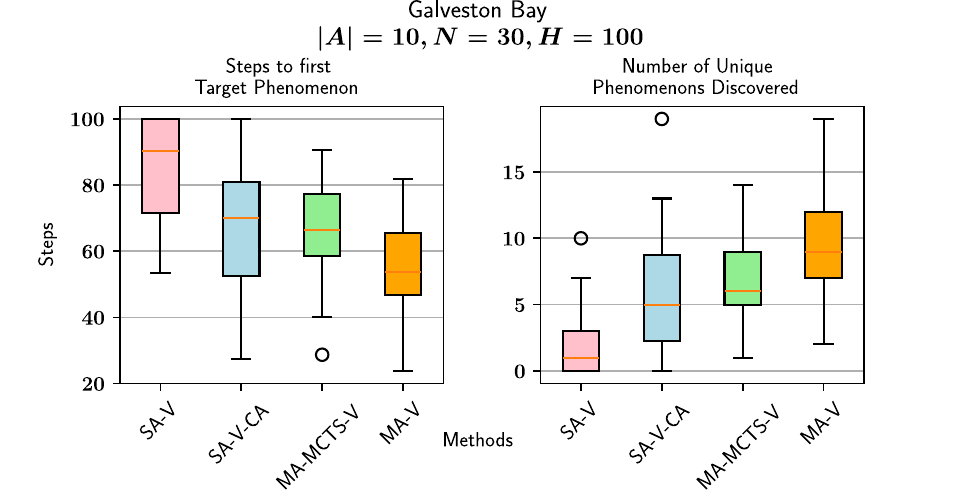}    
    }
    \caption{Scalability experiment over 50 test runs on Galveston Bay with larger number of agents and number of phenomenons.}
    \label{fig:scalability}
\end{figure}

\begin{figure}[ht]
    \centering
    \scalebox{1}{
    \includegraphics[width=0.45\textwidth]{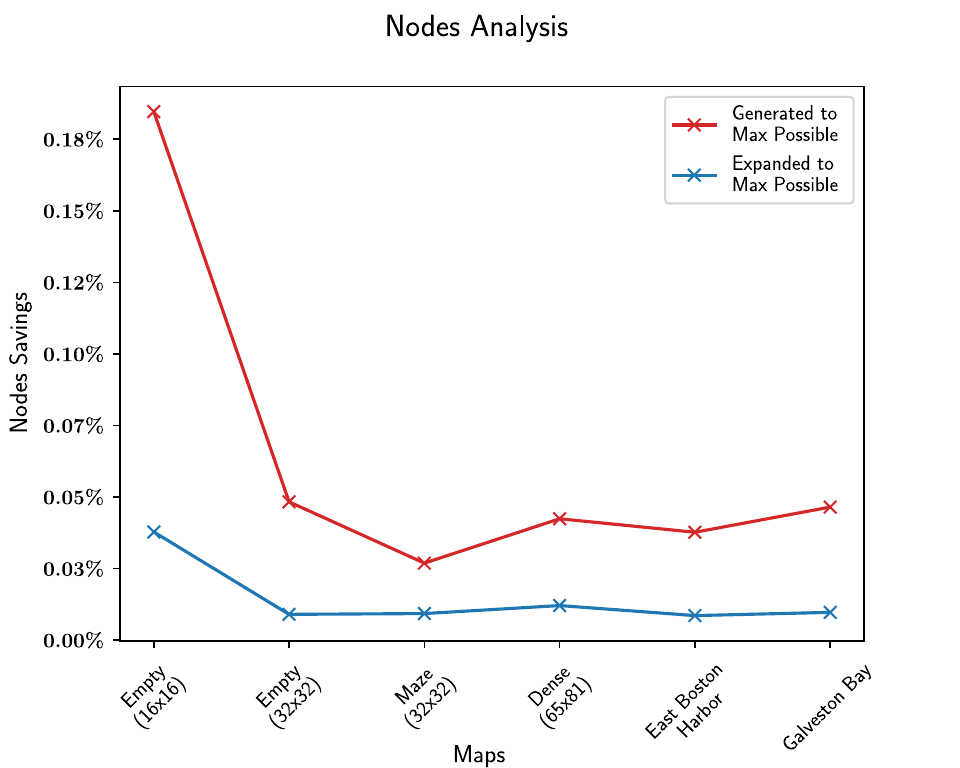}    
    }
    \caption{Ratio of number of A* search states generated and expanded compared to the maximum possible search states.}
    \label{fig:ratio_to_max}
\end{figure}

Figure \ref{fig:phenomenons_discovered_mapf} addresses \textbf{Q1}, showing that across different MAPF maps, our proposed approach successfully discovers more phenomena of interest within the same mission duration while avoiding collisions with other agents. A similar trend is also observed when we ran our approach on the real-world bathymetry datasets as shown in Figure \ref{fig:phenomenons_discovered_real}. Figure \ref{fig:step_to_first_gp_mapf} addresses \textbf{Q2}, demonstrating that agents utilizing our approach encounter their first unique phenomenon of interest sooner compared to the alternative methods in MAPF maps. This suggests that agents effectively leverage observations from their counterparts to explore different map areas thereby leading to efficient map exploration. Similar trends were also observed on the real-world bathymetry datasets as well as shown in Figure \ref{fig:step_to_first_gp_real}. Motivated by realistic scenarios, figure \ref{fig:scalability} showcases the performance of the proposed method compared to the baselines upon scaling the number of agents and phenomenons. It can be seen that when agents utilize our approach, the performance improvement is maintained regardless of the scale of the problem at hand. Figure \ref{fig:ratio_to_max} addresses \textbf{Q3}, illustrating that our approach generates and expands only a fraction of the maximum possible search states, particularly noticeable in larger maps including real-world bathymetry datasets where this ratio approaches zero. This indicates that our approach efficiently decides optimal paths for agents within communication range, leveraging our proposed heuristic to significantly accelerate computation by minimizing the need for complex, compute-intensive coupled rewards. 

To further validate the efficacy of our approach, we conducted experiments on real hardware involving multiple Turtlebots navigating an enclosed space with simulated measurement fields \footnote{\url{https://info-mapf-mers.csail.mit.edu}}.

\section{Conclusion and Future Work}

In conclusion, this paper presents Multi-Agent Vulcan, a novel approach designed for information-guided multi-agent path finding problem where multiple agents are tasked to identify as many phenomena of interest as possible within a limited mission duration. We pose this as a receding horizon MA-POMDP problem in a limited communication setting. By decoupling multi-agent search into multiple single-agent search like MAPF we define an admissible heuristic in the reward space that allows us to leverage informed search methods like A* to find optimal collision-free paths for the agents. We compare our approach against existing adaptive sampling methods inspired by \cite{ayton2017} over multiple MAPF maps and realistic scenarios derived from existing bathymetry datasets. We further validate the advantage of our approach on real-hardware testbeds that used a team of turtlebots to navigate a given environment with simulated measurement fields.

Our method demonstrates significant improvements in our experiments, however a primary challenge remains the compute-intensive estimation of the expected multi-agent information gain ($g(s)$), especially as the number of agents increases. Future work focuses on formulating an efficient estimator for $g(s)$, estimating it through sample-based methods instead of using exact computation. An efficient estimator would result in significant speed up, and allow us to apply our work to even larger multi-agent groups.

\addtolength{\textheight}{-12cm}
\bibliographystyle{./IEEETransactions}
\bibliography{IEEEReferences}

\end{document}